\documentclass[]{interact}

\usepackage{epstopdf}% To incorporate .eps illustrations using PDFLaTeX, etc.
\usepackage[caption=false]{subfig}% Support for small, `sub' figures and tables

\usepackage[T1]{fontenc}
\usepackage{enumerate}
\usepackage{amssymb}
\usepackage{amsmath}

\usepackage[numbers,sort&compress]{natbib}% Citation support using natbib.sty
\bibpunct[, ]{[}{]}{,}{n}{,}{,}% Citation support using natbib.sty
\makeatletter% @ becomes a letter
\def\NAT@def@citea{\def\@citea{\NAT@separator}}% Suppress spaces between citations using natbib.sty
\makeatother% @ becomes a symbol again

\theoremstyle{plain}% Theorem-like structures provided by amsthm.sty
\newtheorem{theorem}{Theorem}[section]

\newtheorem{proposition}[theorem]{Proposition}

\theoremstyle{definition}
\newtheorem{definition}[theorem]{Definition}
\newtheorem{example}[theorem]{Example}

\theoremstyle{remark}
\newtheorem{remark}{Remark}

\def \R {\mathbb{R}}
\def \C {\mathbb{C}}

\def \e {\varepsilon}

\def \K {\mathbb{K}}
\def \G {\mathbb{G}}

\def \g {\mathfrak{g}}

\def \de {\partial}
\def \Exp {\mathrm{Exp}}

\def \Ad {\mathrm{ad}\,}

\def \longto {\longrightarrow}

\def \lan {\langle}
\def \ran {\rangle}

\begin{document}

%\articletype{ARTICLE TEMPLATE}% Specify the article type or omit as appropriate

\title{On the Baker-Campbell-Hausdorff Theorem:
 non-convergence and prolongation issues}

\author{
\name{
Stefano Biagi\textsuperscript{a}\thanks{Corresponding author. Email {andrea.bonfiglioli6@unibo.it}},
Andrea Bonfiglioli\textsuperscript{b}
and Marco Matone\textsuperscript{c}
}
\affil{
\textsuperscript{a}Dipartimento di Ingegneria Industriale e Scienze Matematiche,
 Università Politecnica delle Marche,  via Brecce Bianche, 12, I-60131 Ancona, Italy; \\
\textsuperscript{b}Dipartimento di Matematica, Università di Bologna,
 Piazza Porta San Donato, 5, I-40126 Bologna, Italy;\\
\textsuperscript{c}Dipartimento di Fisica e Astronomia ``G. Galilei'' and Istituto Nazionale di Fisica Nucleare, Università di Padova,
 via Marzolo, 8, I-35131, Padova, Italy.
}
}

\maketitle

\begin{abstract}
 We investigate some topics related to the
 celebrated Baker-Campbell-Hausdorff Theorem: a non-convergence result and prolongation issues. Given a Banach
 algebra $\mathcal{A}$ with identity $I$, and given $X,Y\in \mathcal{A}$, we study the relationship of different issues:
 the convergence of the BCH series $\sum_n Z_n(X,Y)$, the existence of a logarithm of $e^Xe^Y$, and the convergence
 of the Mercator-type series $\sum_n {(-1)^{n+1}}(e^Xe^Y-I)^n/n$ which provides a selected logarithm of $e^Xe^Y$.
 We fix general results and, by suitable matrix counterexamples, we show that various pathologies can occur, among which we provide a non-convergence result for the
 BCH series.
 This problem is related to some recent results, of interest in physics, on closed formulas for the BCH series:
 while the sum of the BCH series presents several non-convergence issues, these closed formulas can provide
 a prolongation for the BCH series when it is not convergent.
 On the other hand, we show by suitable counterexamples that an analytic prolongation of the BCH series can be singular
 even if the BCH series itself is convergent.
\end{abstract}


\begin{amscode}
 Primary: 15A16; 15B99; 40A30.
 Secondary: 34A25.   	
 %40A30: Convergence and divergence of series and sequences of functions
 %34A25: Analytical theory: series, transformations, transforms, operational calculus, etc. [See also 44-XX]
 %22E60
 %17B66;
 %22E05.
\end{amscode}
%%%%%%%%%%%%%%%%%%%%%%%%%%%%%%%%%%%%%%%%%%%%%%%%%%%%
\section{Introduction and motivations}\label{sec:intro}
 The well-known theorem bearing the names of Baker, Campbell and Hausdorff (BCH, in the sequel)
 has pivotal applications both in mathematics and in physics: for instance,
 in the structure theory of Lie algebras  and Lie groups (both finite- and infinite-dimensional), in group theory, in the analysis of linear PDEs, in the theory of ODEs, in control theory,
 in numerical analysis (particularly in geometric integration), in operator theory, in quantum and statistical mechanics, in physical chemistry, in statistical physics and in
 quantum field theories.
 See \cite{AchillesBonfiglioli} or the recent monograph \cite{BonfiglioliFulci} for a list of related references.

  In very recent years, particularly significant in physics has been the derivation of \emph{closed formulas} for the BCH
  series $Z(x,y):=\sum_nZ_n(X,Y)$, when $X$ and $Y$ are operators satisfying specific commutator relations. Such a progress originated in the
   paper \cite{VanbruntVisser} by Van-Brunt and Visser, and it was soon realized that closed BCH formulas admit relevant extensions by introducing simple algorithms, see \cite{Matone:2015wxa}.
In \cite{Matone:2015xaa} it has been shown that there are 13 types of commutator algebras admitting such closed forms for the BCH formula.
Subsequently, closed BCH formulas for the generators of semisimple complex Lie algebras were derived in \cite{Matone:2015oca}, where
an iterative algorithm generalizing the one in \cite{Matone:2015wxa} was also introduced.

 The above results have been applied in covariantizing the generators of the conformal transformations, in providing explicit expressions of the unitary representations of the fundamental group of Riemann surfaces, in the context
 of conformal field theories, see \cite{Matone:2015hja}. Furthermore, the algorithm in \cite{Matone:2015wxa} was applied in investigating the zero-energy states in conformal field theory with sine-square deformation,
 see \cite{Tamura:2017vbx}. Closed BCH formulas have been found for the contact Heisenberg algebra, see \cite{Bravetti:2016qve}.
 Related investigations also concern the recent papers \cite{Foulis:2017hfo, Kimura:2017xxz} on the Zassenhaus formula.\medskip

  Typically, closed BCH formulas can be derived by a formal
  manipulation of the BCH series $Z(x,y)$, often expressed as the logarithm $\ln(e^Xe^Y)$, or via certain integral representations
  for the sum of the series.
  In general, $Z(X,Y)$ coincides with $\ln(e^Xe^Y)$ only for small norms of $X$ and $Y$ (see Proposition \ref{prop.factspos}).
  Since the power expansion of the exponential is everywhere convergent,
  it follows that a closed expression, say $L(X,Y)$, for
  the sum of the series $\sum_n {(-1)^{n+1}}(e^Xe^Y-I)^n/n$ (giving one selected logarithm of $e^Xe^Y$), which turns out to be meaningful (and analytic)
  in a wider region than the set of convergence of that series, will fulfill
  the identity $e^{L(X,Y)}=e^Xe^Y$ by analytic continuation. The latter identity is very often what physicists look for when dealing with BCH.
  Equivalently put, closed formulas for the BCH series obtained in the above way should be referred to as
  \emph{prolongations} of the sum of the series $Z(X,Y)$, when the latter is not convergent, which is often the case.
  Unfortunately, as we will show, the existence and the actual values of $Z(X,Y)$, of $L(X,Y)$, or of their prolongations can be very differently behaved.\medskip

 Since the coefficients $Z_n(X,Y)$ of the BCH series are well posed in any Lie algebra, the problem of the {convergence}
 of the BCH series $\sum_{n=1}^\infty Z_n(x,y)$ is meaningful in any Lie algebra equipped
 with a metric, e.g., in finite-dimensional Lie algebras or in Banach-Lie algebras.\footnote{By a Banach-Lie algebra we mean a
  Banach space $\mathcal{A}$ (over $\R$ or $\C$) endowed with a Lie algebra structure such that $\mathcal{A} \times \mathcal{A}\ni
  (x, y) \mapsto [x, y]\in \mathcal{A}$ is continuous.} The study of the convergence domain of the BCH series has a very long history, tracing back to Hausdorff
 \cite[Section 4]{Hausdorff}, and the determination of the optimal domain of convergence is still an open and trying problem.
 See, e.g.,
 \cite{BiagiBonfiglioliLAMA, BlanesCasas, BlanesCasasOteoRos2, Bourbaki, Casas, CasasMurua, DaySoThompson, Dynkin, Merigot, Michel, NewmanSoThompson, Postnikov, Strichartz, Suzuki, Thompson, Vinokurov}.
 Related references to the convergence domain of the BCH series focused on its continuous counterpart (of great importance in the applications), the so-called Magnus series,
 can be found e.g., in \cite{BlanesCasasOteoRos, Moan, MoanNiesen, MoanOteo}.

 We do not contribute to the study of the optimality of the BCH convergence set, which is a highly nontrivial problem, as shown e.g., in
 \cite{BlanesCasas, CasasMurua, MoanNiesen}.
 Instead, with the use of selected counterexamples,  we limit ourselves to showing that the BCH and logarithmic series $Z(X,Y), L(X,Y)$ can be very differently behaved,
 as far as convergence/divergence are concerned  (see Section \ref{sec:prolongation}). The problem of the prolongability
 of $Z(X,Y)$ (and whether or not this prolongation gives information on the convergence of $Z(X,Y)$) is also studied here, highlighting
 some unexpected pathologies.

% In the setting of infinite-dimensional Banach-Lie algebras, the convergence problem was recently investigated in \cite{BiagiBonfiglioliLAMA},
% by extending to that setting the result on the convergence domain known in the finite-dimensional case (Mérigot \cite{Merigot}; Blanes, Casas \cite{BlanesCasas}).
% The idea of using the (formal) ODE solved by
% $t\mapsto Z(tx,ty)$ has a crucial role in \cite{BiagiBonfiglioliLAMA, BlanesCasas, Merigot}; this ODE-approach traces back to Poincaré himself:
% see \cite[equation (7), p.\,248]{Poincare}, where an ODE solved by $t\mapsto Z(x,ty)$ was first discovered. We shall use again an ODE
% technique in Section \ref{sec:prolongation}, while anticipating here that this versatile approach has already proved useful in other
% contexts related to BCH (see e.g., the Zassenhaus formula in the paper \cite{ArnalCasasChiralt} by Arnal, Casas, Chiralt, or the prolongation problem for the BCH series
% in the paper \cite{Eggert} by Eggert).

 \medskip

 We would like to highlight some physics applications: there is a basic reason why the convergence issue of the prolongation of the BCH series is a central question, not only by a purely technical point of view,
 but also of considerable mathematical and physical interest. It is known that the BCH Theorem, in its various forms, plays a key role
 in quantum mechanics, quantum field theories, including their path integral formulation, and statistical physics. Main examples concern the Trotter product formula \cite{Strocchi}
 $$
 \exp(A+B)=\lim_{n\to\infty}(\exp(A/n)\exp(B/n))^n,
 $$
 the Magnus expansion, and its generalizations \cite{Bauer}. To understand this, recall that, for example, the transition amplitude $\langle q',t'|q,t\rangle$ that leads to the Dirac-Feynman path integral
 has the form
 \begin{equation}\label{amplitude}
 \langle q',t'|q,t\rangle = e^{-{\frac{i}{\hbar}}H(t'-t)}\delta(q'-q),
 \end{equation}
 where $H=-\frac{\hbar^2}{2m}\, \Delta +V(q)$ is the Hamiltonian, $\Delta$ the Laplace-Beltrami operator and $V(q)$ the potential. A similar expression concerns
 the extension to quantum field theory, where now
 the Dirac tempered distribution is replaced by the functional Dirac distribution.
 In most theories, the path integral formulation is treated as a
 power expansion in the coupling constants. An outstanding problem is that such expansions yield divergent asymptotic series.
 In recent years a new approach,  based on \'Ecalle resurgence theory \cite{Ecalle}\footnote{See
 \cite{Sun:2018ceu} for a recent introduction to the related mould calculus applied to the BCH formula.}, has been developed
 \cite{Basar:2017hpr, Serone:2017nmd}. The main idea is to use transseries expansions, which are faithful and
 unambiguous representations of observables. In such a construction, the analytic continuation plays the fundamental role.
 On the other hand, as illustrated by \eqref{amplitude}, it is clear that the problem of analytic continuation
 translates into a problem of analytic continuation of the BCH formula.\bigskip

  Finally, we describe the plan of the paper.\medskip

  In Section \ref{sec:notation} we
  introduce the precise notation used in the paper along with the statements of the main results.
  We  show that the convergence of $Z(X,Y)$ is totally independent of the existence of $\ln(e^Xe^Y)$, and
  even when both $Z(X,Y)$ and $\ln(e^Xe^Y)$ exist, their values  can actually be different.
 More generally, given a Banach
 algebra $\mathcal{A}$ with identity $I$, and given $X,Y\in \mathcal{A}$, we study three different issues:
 the convergence of the BCH series, the existence of a logarithm of $e^Xe^Y$, and the convergence
 of the Mercator-type series $\sum_n \tfrac{(-1)^{n+1}}{n}(e^Xe^Y-I)^n$ which provides a selected logarithm of $e^Xe^Y$.
 We fix general results and, by suitable
 matrix algebra counterexamples, we show that various pathologies can occur; we provide our main
 \emph{non-convergence} result for the BCH series, for the proof of which a simple Lie-algebraic argument is used.\medskip

  In Section \ref{sec:non-convergence} we give the proofs of the results of Section \ref{sec:notation}.\medskip

  In Section \ref{sec:prolongation} we exhibit some subtle pathologies which can intervene
  in the problem of the \emph{prolongation} of the BCH series.  For instance,
  we show examples where $Z(X,Y)$ can be analytically prolonged to a function $P(X,Y)$ but:
\begin{enumerate}
  \item $P(X,Y)$ has a singularity at $(X_0,Y_0)$ but $Z(X_0,Y_0)$ converges
  (Example \ref{esempio.bello.stefano});

 \item $P(X,Y)$ is everywhere defined on $\g\times \g$ (where $\g$ is a suitable matrix Lie algebra), whereas $Z(X,Y)$ is somewhere non-convergent (Example \ref{exa.eggertiano}).
\end{enumerate}
 The example in (2) above is obtained by combining
 an abstract result in \cite{Eggert}, together with a class of examples of
 Lie algebras of vector fields contained in \cite{BonfLancCPAA}.
 The problem of the convergence of $Z(X,Y)$ in Lie algebras of vector fields
 is of independent interest in the analysis of H\"ormander operators
 (see e.g., \cite{BiagiBonfiglioliCCM, BiagiBonfiglioliPLMS, BonfiMEDIT, BonfLancCPAA}),
 and we shall return to it in a future investigation.

%%%%%%%%%%%%%%%%%%%%%%%%%%%%%%%%%%%%%%%%%%%%%%%%%%%%%%%%%%%%%%%%%%%%%%%%%%%%%%%%%%%%%%%%%%%%%%%%%%%%%%%%%%%%%%
 \section{Notations and main results}\label{sec:notation}
  In its most basic algebraic form, the BCH Theorem ensures that, in the associative algebra
 $\mathbb{K}\langle\!\langle x,y \rangle\!\rangle$ of the formal power
 series in two non-commuting indeterminates $x$ and $y$ over a field $\mathbb{K}$ of characteristic zero, one has
 $e^xe^y=e^{Z(x,y)}$, where $Z(x,y)$ can be expressed as a series of Lie polynomials
 $$Z(x,y)=x+y+\tfrac{1}{2}[x,y]+\tfrac{1}{12}([[x,y],y]+[[y,x],x])-\tfrac{1}{24}[x,[y,[x,y]]]+\cdots .$$
 To be precise, in the sequel we consider the series $\sum_{n=1}^\infty Z_n(x,y)$ that can be obtained from $Z(x,y)$ by grouping together the
 Lie polynomials of degree $n$ in $x$ and $y$, i.e.,
 $$\text{$Z_1(x,y):=x+y$,\quad
 $Z_2(x,y):=\tfrac{1}{2}[x,y]$,\quad
 $Z_3(x,y):=\tfrac{1}{12}([[x,y],y]+[[y,x],x])$,\quad
 etc.}  $$
 This is the so-called \emph{homogeneous} (presentation of the) BCH series.
 Throughout, when a BCH series is concerned, we always
 tacitly understand the homogeneous one.

 More explicitly, once it is known that $Z(x,y)$ is a Lie series, the $Z_n$'s can be explicitly written,
 via the Dynkin-Specht-Wever Lemma (as in \cite[Sec.\,3.3.2]{BonfiglioliFulci}), under the following well-know (Dynkin) presentation\footnote{When
 $j_k=0$ (thus $i_k\neq 0$) the associated summand in \eqref{notation.dynkin1}
 is understood to end with
 $(\Ad x)^{i_k-1}(x)$.}
 \begin{gather}\label{notation.dynkin1}
 \begin{split}
 &Z_n(x,y):=\\
 &\qquad \frac{1}{n}\sum_{k=1}^n \frac{(-1)^{k+1}}{k}\,
 \!\!\!\!\!\sum_{\substack{(i_1,j_1),\ldots,(i_k,j_k)\neq (0,0)\\
 i_1+j_1+\cdots+i_k+j_k=n}}\!\!\!\!\!
 \frac{(\Ad x)^{i_1}(\Ad y)^{j_1}\cdots (\Ad x)^{i_k}(\Ad
 y)^{j_k-1}(y)}{i_1!\,j_1!\,\cdots\,i_k!\,j_k!}\,.
\end{split}
\end{gather}
 The convergence of the BCH series $\sum_{n=1}^\infty Z_n(x,y)$ in the usual\footnote{See e.g., \cite[Section 2.3.3]{BonfiglioliFulci}.}
 metric topology of $\mathbb{K}\langle\!\langle x,y \rangle\!\rangle$
 is a trivial consequence of the increasing degrees of the $Z_n$'s.

 As this will be relevant throughout the paper, we review that (in the algebraic setting of $\mathbb{K}\langle\!\langle x,y \rangle\!\rangle$)
 the series $Z(x,y)$ is uniquely given by $\ln(e^xe^y)$, where
\begin{equation}\label{logaserie}
    \ln(W)=\sum_{k=1}^\infty \frac{(-1)^{k+1}}{k}\,(W-I)^k,
\end{equation}
  for any formal power series $W\in \mathbb{K}\langle\!\langle
 x,y \rangle\!\rangle$ whose zero-degree term is equal to the identity $I$ of $\K$.
 Since the power series $\sum_{k=1}^\infty {(-1)^{k+1}}\,z^k/k$ is usually called the `Mercator series',
 in order to avoid ambiguities in situations where other logarithms can be meaningful
 (as in $\C$ or in matrix algebras), we introduce once and for all a selected notation for what we shall mean by $\ln(e^xe^y)$ in more general settings:
\begin{equation}\label{mercatorlogexp}
\begin{split}
 &L(x,y):=\sum\nolimits_{n} L_n(x,y), \\
 &\qquad\text{where $L_n(x,y):=\frac{(-1)^{n+1}}{n}\,(e^xe^y-I)^n$ for any $n\in \mathbb{N}$.}
 \end{split}
\end{equation}
 With a little abuse, we say that $\sum_{n} L_n(x,y)$ is the \emph{Mercator series} (a shorthand of `Mercator series for $\ln(e^xe^y)$');
 we also say that its sum $L(x,y)$ is the Mercator logarithm of $e^xe^y$, and (when there is no risk of confusion), we may write $\ln(e^xe^y)$
 in place of $L(x,y)$.
 The following identity
\begin{equation}\label{daunaassocloc}
 Z(x,y)=L(x,y)\qquad \text{(also written as $Z(x,y)=\ln(e^xe^y)$)}
\end{equation}
 is clearly equivalent to
  $e^{Z(x,y)}=e^xe^y$, another identity in the formal-power-series setting of
  $\mathbb{K}\langle\!\langle x,y \rangle\!\rangle$
  (this equivalence being not always true in other settings).
  \medskip
%  Incidentally, as it will be important in the sequel,
%  \eqref{daunaassocloc} justifies
%  the associativity property of the map $(a,b)\mapsto Z(a,b)$, i.e.,
%\begin{equation}\label{daunaassoclocASSSS}
% Z(a,Z(b,c))=Z(Z(a,b),c),
%\end{equation}
% holding true for any $a,b,c\in \mathbb{K}\langle\!\langle x,y \rangle\!\rangle$
% whose zero-degree terms are null.\medskip

 All these facts are so well established in the BCH folklore that one may forget that the following four issues,
 though simple to solve in the formal-power-series setting,
 may be highly non-trivial if one is working outside $\mathbb{K}\langle\!\langle x,y \rangle\!\rangle$:
 \begin{enumerate}
   \item[(I)] the convergence of the Mercator series $\sum_{n=1}^\infty (-1)^{n+1}\,(e^xe^y-I)^n/n$;
   \item[(II)] the convergence of the BCH series $\sum_{n=1}^\infty Z_n(x,y)$;
   \item[(III)] the identity $e^{Z(x,y)}=e^xe^y$, or (generally) the existence of a logarithm of $e^xe^y$;
   \item[(IV)] the equality of the sums $L(x,y)$ and $Z(x,y)$ of the series in (I) and (II).
 \end{enumerate}

 Apart from the algebraic framework of $\mathbb{K}\langle\!\langle x,y \rangle\!\rangle$,
 another kettle of fish is the validity of these four issues
 when $x$ and $y$ belong to more specific topological spaces, as in the case of matrix algebras.
 One of the aims of this paper is to study these problems in a wide framework, that of the \emph{Banach algebras}: we fix some positive general results, and
 (with the use of selected counterexamples), we show that problems (I)-to-(IV) can be very differently behaved.
  Indeed, we shall see that many pathological facts do occur even in the simple case of matrix algebras:
 for example, the BCH series may converge, whereas the Mercator series
 may not; or viceversa;
 or they can be both convergent but with different sums; or they can be both non-convergent, but $e^xe^y$
 may yet admit a logarithm, i.e., some  solution $V$ of $e^V=e^xe^y$.

 As a result, this will point out some inaccuracies, sometimes appearing in the literature, due to a formal application of the BCH Theorem. Indeed,
 any formal manipulation (resemblant to what is allowed in $\mathbb{K}\langle\!\langle x,y \rangle\!\rangle$ with its very simple topology)
 of the BCH series $Z(x,y)$ or of its companion Mercator series $\ln(e^xe^y)$ will invariably lose track of all the mentioned pathologies, especially in convergence issues.\medskip

 In view of the applications, the physics community has paid much attention to the convergence of the BCH series, as already described. However,
  some ambiguity occasionally arises from a formal manipulation of the BCH series.
 %The first main result of this paper concerning the \emph{non-convergence} of the BCH series is aimed to shed light on this potential ambiguity, which we now closely describe.
 Indeed, when dealing with the BCH Theorem in physics applications, one often meets with the  identity
 %(see e.g., \cite[eq.\,(1)]{VanbruntVisser})
\begin{equation}\label{identitafalsa}
 \ln(e^xe^y)=x+y+\tfrac{1}{2}[x,y]+\tfrac{1}{12}([[x,y],y]+[[y,x],x])+\cdots.
\end{equation}
 Unfortunately, while \eqref{identitafalsa} is certainly true in the formal-series setting of $\mathbb{K}\langle\!\langle
 x,y \rangle\!\rangle$, and it is true in any Banach algebra
 provided that $x$ and $y$ are sufficiently close to zero,
 it can be dramatically false otherwise. Roughly put, this
 is due to the fact that (after an expansion of $\ln(e^xe^y)$ in its Mercator series),
 the terms on any of the two sides of \eqref{identitafalsa}
 are obtained from the terms on the other side by reordering and associating, which are
 not harmless procedures   even in the case of real-valued series.

 Another observable issue of \eqref{identitafalsa} lies in the
 logarithm, in that, while it is uniquely given by \eqref{logaserie} in $\mathbb{K}\langle\!\langle
 x,y \rangle\!\rangle$, in special Banach algebras \eqref{logaserie} may not be the optimal choice:
 consider, for instance, the case of matrix algebras, where a more efficient $\ln$-function can be defined, for many classes of matrices, through the
 Jordan decomposition. Since we consider the general case of Banach algebras, we are compelled to unambiguously choose
 what we mean by the logarithm, which we now do.

 In view of the fact that the BCH coefficients $Z_n$ are constructed via the Mercator series \eqref{logaserie}
 (as it is also visible from the factors $(-1)^{k+1}/k$ in \eqref{notation.dynkin1}), it appears that \eqref{logaserie} is the most natural choice
 if one wants to give a sufficiently comprehensive analysis of our problem, applicable to
 Banach algebras.
 Furthermore, \eqref{logaserie} makes unambiguous sense in any Banach algebra $\mathcal{A}$, if we mean by $I$
 the identity element of $\mathcal{A}$. For this reason, as is frequently done for operator algebras, here and in the sequel we
 adopt the following definition.
\begin{definition}\label{defi.mercator}
 Let $\mathcal{A}$ be a Banach algebra, i.e., a triple $(\mathcal{A},*,\|\cdot\|)$
 where $(\mathcal{A},*)$ is a unital associative algebra (with identity denoted by $I$), and
 $(\mathcal{A},\|\cdot\|)$
 is a (real or complex) Banach space, where the norm $\|\cdot\|$ is compatible with the multiplication,
 i.e., $\|x*y\|\leq \|x\|\,\|y\|$ for
 any $x,y\in \mathcal{A}$.

 Given $W\in \mathcal{A}$, $\ln(W)$ will denote the sum of the Mercator series in \eqref{logaserie},
 when this series converges in the metric space $\mathcal{A}$.
 The function  $\exp:\mathcal{A}\to \mathcal{A}$
 is defined via the usual series $\sum_{k=1}^\infty {W^k}/{k!}$
 (and denoted indifferently by $\exp(W)$ or $e^W$),
 this series being absolutely\footnote{We say that a series $\sum_n a_n$ in $\mathcal{A}$ is
 absolutely convergent if $\sum_n\|a_n\|<\infty$.}
 %If a series is convergent, but we do not know if this is also absolutely convergent, we shall occasionally say that it is simply convergent.}
 convergent for any $W\in \mathcal{A}$.
 In what follows, given $W\in \mathcal{A}$, we say that $V\in \mathcal{A}$ is a logarithm of $W$ if $e^V=W$.

 Finally, given any $x,y\in \mathcal{A}$, the notations in \eqref{mercatorlogexp}
 will be applied for the Mercator series $\sum_n L_n(x,y)$ (be it convergent or not) and
 for its sum $L(x,y)$, occasionally also denoted by $\ln(e^xe^y)$.
\end{definition}
 With these definitions at hand,
 identity \eqref{identitafalsa}
 is the equality of the sums of the Mercator and BCH series, a fact which may easily fail to be true;
 luckily, some positive results are available for identity \eqref{identitafalsa} to hold,
 as the following result ensures, belonging to the BCH folklore.\footnote{For the sake of completeness, the proof of Proposition \ref{prop.factspos}
 (based on some re-arranging argument on absolutely convergent series, and on
 analytic-function theory in Banach algebras)
 is sketched in the Appendix, as it is not so easy to locate in the literature.
}
\begin{proposition}\label{prop.factspos}
 Let $\mathcal{A}$ be a Banach algebra, and let $x,y\in \mathcal{A}$.
 Then:
 \begin{enumerate}
   \item[\emph{(a)}]
   If $\|x\|+\|y\|<\ln 2$, the Mercator series $\sum_n L_n(x,y)$ and the
   BCH series $\sum_n Z_n(x,y)$ are both absolutely convergent; moreover,
   the sums of their series are equal, and
   \eqref{identitafalsa} holds true \emph{(}$\ln(e^xe^y)$ meaning the sum of the Mercator series\emph{)}.\medskip

   \item[\emph{(b)}]
 If the Mercator series $\sum_n L_n(x,y)$ converges in $\mathcal{A}$
 \emph{(}without any knowledge on its absolute convergence\emph{)},
 then its sum $L(x,y)$ is a logarithm of $e^xe^y$.\medskip

   \item[\emph{(c)}]
 The same statement as in \emph{(b)} is valid for the  BCH series.
 \end{enumerate}
\end{proposition}
 \noindent Problems for \eqref{identitafalsa} soon arise for non-small $\|x\|+\|y\|$, as shown in the next example.
\begin{example}\label{examBiagi}
 Let $\mathcal{M}=M_2(\R)$ denote the usual normed algebra of the real $2\times 2$ matrices, and let us consider (for $v\in \R$) the matrices
\begin{equation}\label{CESbiagi}
 x=x(v):=\left(
       \begin{array}{cc}
         -v & 0 \\
         0 & -2v \\
       \end{array}
     \right)\quad \text{and}\quad
     y:=\left(
       \begin{array}{cc}
         0 & 1 \\
         0 & 0 \\
       \end{array}
     \right).
\end{equation}
 In Section \ref{sec:non-convergence} we shall prove that:
\begin{enumerate}[(i)]
  \item the Mercator series for $\ln(e^{x(v)}e^y)$  converges in $\mathcal{M}$ if and only if $v\geq -\ln\sqrt2$;
  \item the BCH series $\sum_{n}Z_n({x(v)},y)$ converges in $\mathcal{M}$ if and only if $|v|<2\pi$;
  \item there exists a logarithm of $e^{x(v)}e^y$ for every $v\in \R$; this is given e.g., by
 $$Z(v):=\left(
       \begin{array}{cc}
         -v & \psi(v)\\
         0 & -2v \\
       \end{array}
     \right), $$
 where $\psi(v):=\frac{v}{1-e^{-v}}$ is Todd's function (it is understood
 that $\psi(0)=1$).
\end{enumerate}
\end{example}
 Keeping in mind Proposition \ref{prop.factspos} and Example \ref{examBiagi}, in Section
 \ref{sec:non-convergence} we shall prove the following result, investigating in general Banach algebras problems (I)-to-(IV) previously considered for the
 formal-power-series algebra
 $\mathbb{K}\langle\!\langle x,y \rangle\!\rangle$: we provide the mutual relationships of (I)-to-(IV), and we point out the involved pathologies
 for \eqref{identitafalsa} to hold:
\begin{proposition}\label{pro.mutuiqalrelat}
 Given $x,y$ in a Banach algebra $\mathcal{A}$, consider the problems:
 \begin{enumerate}
   \item[\emph{(I)}] the Mercator series $\ln(e^xe^y)=\sum_{n} L_n(x,y)$ is convergent in $\mathcal{A}$;
   \item[\emph{(II)}] the BCH series $\sum_{n}^\infty Z_n(x,y)$ is convergent in $\mathcal{A}$;
   \item[\emph{(III)}] there exists a logarithm of $e^xe^y$, i.e., $V\in \mathcal{A}$ fulfilling the identity $e^V=e^xe^y$;
   \item[\emph{(IV)}] the sums of the Mercator and BCH series are equal.
\end{enumerate}
 Then the following facts hold true:
\begin{enumerate}
 \item[1.] \emph{(II)} is sufficient to \emph{(III)}, but not necessary.

 \item[2.]
 \emph{(I)} is sufficient to \emph{(III)}, but not necessary.

 \item[3.]
 \emph{(I)} and \emph{(II)} are independent of each other.

 \item[4.]
 \emph{(I)}, \emph{(II)}, \emph{(III)} may all be false.

 \item[5.]
 \emph{(I)} and \emph{(II)} may hold true, but \emph{(IV)} can be false.
 \end{enumerate}
\end{proposition}
% In physics applications, problem (III) seems to be of the greatest relevance, and, in comparison to the other problems,
% it is much more frequently successful: this is the case of matrix algebras, where a logarithm of $e^xe^y$
% can be defined in many\footnote{Actually, if one is working with complex matrices, since $e^xe^y$ is always invertible,
% it always admits a logarithm.} interesting situations via the Jordan decomposition rather than via \eqref{logaserie}.
% In general Banach algebras, a logarithm of $e^xe^y$ can exist independently of the existence of $L(x,y)$ or of $Z(x,y)$,
% whose associated series may be divergent (see Example \ref{examBiagi}).
% \noindent Keeping in mind Proposition \ref{prop.factspos} and Example \ref{examBiagi}, we are in a position to
% prove Proposition \ref{pro.mutuiqalrelat} (numbers are tacitly related to the five statements in the thesis of
% the proposition):
 The non-convergence result of the BCH series contained in Example \ref{examBiagi}-(ii) is closely related to
 some recent classes of Lie algebras of interest in physics (see \cite{Lo, VanbruntVisser}).
 In \cite{VanbruntVisser}, Van-Brunt and Visser consider the case of two operators $X$ and $Y$ with
 the commutator relation (with scalar $u,v,c$)
\begin{equation}\label{commuvisser}
 [X,Y]=u\,X+v\,Y+c\,I,
\end{equation}
 where $I$ commutes with both $X$ and $Y$. When $u=v=0$, this comprises the Heisenberg case $[P,Q]=-i\hbar I$ and the
 creation-annihilation commutator $[a,a^\dag]=I$.
 We observe that our Example \ref{examBiagi} falls in this class: indeed, if $X$ and $Y$ are respectively
 given by the matrices $x$ and $y$ in \eqref{CESbiagi}, then
 \eqref{commuvisser} holds true with $u=c=0$ (and any $v$).

 Via some formal and tricky manipulation of the BCH series
 based on a (formal) integral representation for its sum (due to Richtmyer and Greenspan, \cite{RichtmyerGreenspan}), in \cite{VanbruntVisser}
 it is shown that this integral representation, under the assumption \eqref{commuvisser}, is equal to
\begin{equation*}%\label{vanbruntvisserformula}
    X+Y+f(u,v)[X,Y],\qquad \text{where $f(u,v)=\frac{ue^u(e^v-1)-ve^v(e^u-1)}{uv(e^u-e^v)}$.}
\end{equation*}
 As the derivation of this object results from
 the BCH series, it seems to lead to a \emph{closed formula for the BCH series}, as they are usually referred to
 in the physics literature.
 Unfortunately, in general $X+Y+f(u,v)[X,Y]$ cannot be claimed to be the
 sum of the BCH series: indeed, there are suitable choices of $u,v,c$ and of $X,Y$
 which give sense to $X+Y+f(u,v)[X,Y]$ but for which the BCH series is non-convergent:
 namely, take $u=c=0$, $|v|\geq 2\pi$  and $X=x(v),Y=y$ in Example \ref{examBiagi}.

 Moreover, even the Mercator logarithm $\ln(e^Xe^Y)$ may be different from
 $X+Y+f(u,v)[X,Y]$. Indeed, if we take $u=c=0$, $v\leq -\ln\sqrt2$, $X=x(v)$ and $Y=y$ in Example \ref{examBiagi}, then
 the Mercator series for $\ln(e^Xe^Y)$ is not convergent, whereas $X+Y+f(0,v)[X,Y]$
 is perfectly meaningful.

 As a consequence, the following identities contained in \cite{VanbruntVisser}
\begin{equation}\label{vanbruntvisserformulasba}
 Z(X,Y)=\ln(e^Xe^Y)=X+Y+f(u,v)[X,Y]
\end{equation}
 must be read, in terms of the BCH series and of the Mercator logarithm, as follows:
 the far right-hand side, say $V$, of \eqref{vanbruntvisserformulasba}
  is \emph{a prolongation} (for the values of $u,v$ in the domain of $f$) both of the BCH series $Z(X,Y)$ and of the
 Mercator series for $\ln(e^Xe^Y)$
 when these series are not convergent;
 %(two prolongation results of their own interest);
 moreover,  $V$ is a logarithm of $e^Xe^Y$, i.e.,
\begin{equation*}%\label{vanbruntvisserformula3}
    \exp\big(X+Y+f(u,v)[X,Y]\big)=e^Xe^Y.
\end{equation*}
% Since Richtmyer and Greenspan's integral formula for the sum of the BCH series
% is valid when $X$ and $Y$ are adequately small, formula \eqref{vanbruntvisserformula3}
% holds true when $X$ and $Y$ are sufficiently close to the origin of any Banach algebra, provided that
% they satisfy \eqref{commuvisser}.\medskip
 However, the existence of a prolongation of $Z(X,Y)$
 does \emph{not} imply the convergence of $Z(X,Y)$ and, viceversa, if a prolongation of the BCH series is singular
 at  $(X_0,Y_0)$ this does \emph{not} imply  that $Z(X_0,Y_0)$ is non-convergent (see Section \ref{sec:prolongation}).\medskip

 In order to clarify the possible non-convergence of the BCH series repeatedly mentioned above, we now state the following result, proved in Section \ref{sec:non-convergence}:
\begin{theorem}[\textbf{A non-convergence result for the BCH series}]\label{mainteo.nonconv}
 Let $\mathcal{A}$ be a Banach algebra \emph{(}or, more generally, a Banach-Lie algebra\emph{)} over
 $\R$ or $\C$. Assume that there exist $X,Y\in \mathcal{A}$
 \emph{(}with $Y\neq 0$\emph{)} and a scalar $v$ such that
 $[X,Y]=vY$.

 Then the BCH series $(X+Y)+\sum_{n=2}^\infty Z_n(X,Y)$
 coincides with the series
 $$(X+Y)+\sum_{n=1}^\infty\frac{(-v)^nB_n}{n!}\,Y, $$
 where the $B_n$'s are the Bernoulli numbers, i.e., the rational numbers uniquely determined
 by the generating function $\frac{z}{e^z-1}=\sum_{n=0}^\infty \frac{B_n}{n!}\,z^n$.
 Thus the BCH series $\sum_{n=1}^\infty Z_n(X,Y)$
 converges in $\mathcal{A}$ if and only if
 $|v|<2\pi$, and in this case the sum of the series is
\begin{equation}\label{vanbruntvisserformula2}
 X+\psi(v)\,Y, \quad \text{where $\psi(v):=\frac{v}{1-e^{-v}}$ is Todd's function}.
\end{equation}
 In particular, if $|v|\geq 2\pi$ the BCH series  $\sum_{n=1}^\infty Z_n(X,Y)$
 is not convergent.
 \end{theorem}
 The case $Y=0$ in Theorem \ref{mainteo.nonconv} is trivial, since $Z(X,0)=X$ is always convergent.

 It may be of interest to observe that, in proving Theorem \ref{mainteo.nonconv}, we shall not make use
 of the integral representation of the $Z_n$'s
 by Richtmyer and Greenspan,  \cite{RichtmyerGreenspan} (as is done in \cite{VanbruntVisser}), but only of
 a simple Lie-algebra argument.\medskip

%\begin{remark}
% The matrices in \eqref{CESbiagi} are only one of the infinitely-many
% possible examples of $2\times 2$ real matrices $X,Y$ such that $[X,Y]=vY$. Another example
% is studied in \cite[Sec.\,V]{VanbruntVisser}, namely
%\begin{equation}\label{Visserchoicematrices}
% A=\left(
%       \begin{array}{cc}
%         v/2 & 1 \\
%         0 & -v/2 \\
%       \end{array}
%     \right)\quad \text{and}\quad
%     B=\left(
%       \begin{array}{cc}
%         0 & 1 \\
%         0 & 0 \\
%       \end{array}
%     \right).
%\end{equation}
% If $A,B$ are as above,
% it follows from Theorem \ref{mainteo.nonconv}
% that the BCH series
% $\sum_{n}Z_n(A,B)$ converges iff $|v|<2\pi$, which is the same
% condition that we had for our matrices \eqref{CESbiagi}.
% This is not surprising, since the proof of Theorem \ref{mainteo.nonconv}
% is based on a purely Lie-algebraic argument, and thus
% it only depends on the commutator relations of $A,B$.
% Much differently,
% the condition for the existence of the Mercator logarithm $\ln(e^Xe^Y)$ heavily depends
% on the matrices representing $X$ and $Y$ (since it is related to the spectrum of $X,Y$ as operators); for instance, one can prove that
% the Mercator logarithm $\ln(e^{A}e^{B})$ associated with \eqref{Visserchoicematrices} exists iff
% $|v|\leq \ln 4$,
% a quite different result if compared to (i) in Example \ref{examBiagi}.\medskip
%\end{remark}

  Once it is clear that
  the prolongations of the maps $(X,Y)\mapsto Z(X,Y), L(X,Y)$
  must not be confused with the convergence of the associated series, we think it is worthwhile to have
  some counterexamples at hand, showing the total independence of prolongation and convergence: these are provided in Section \ref{sec:prolongation}, where
  we shall prove the next result.
\begin{example}\label{esempio.bello.stefano}
  Let $\mathcal{M}=M_2(\C)$ denote the usual normed algebra of the complex $2\times 2$ matrices.
  In $\mathcal M$ we consider the matrices
\begin{equation}\label{CESbiagibello}
 X(\alpha):=\left(
       \begin{array}{cc}
         -\alpha & 0 \\
         0 & -2\alpha \\
       \end{array}
     \right)\quad \text{and}\quad
     Y(\beta):=\left(
       \begin{array}{cc}
         0 & \beta\,(\beta-2\pi\,i) \\
         0 & 0 \\
       \end{array}
     \right).
\end{equation}
 Then the BCH series $Z(\alpha,\beta):=Z(X(\alpha),Y(\beta))$  converges if and only if
 $(\alpha,\beta)$ is in
\begin{equation}\label{CESbiagibello2}
 D:=
 \big\{(\alpha,\beta)\in \C^2:\,|\alpha|<2\pi,\,\,\beta\notin\{0,2\pi\,i\}\big\}\cup
 (\C\times \{0,2\pi\,i\}),
\end{equation}
 and the sum is given by
\begin{equation}\label{sommosa}
 Z(\alpha,\beta)=\left\{
                   \begin{array}{ll}
\left(\begin{array}{cc}
         -\alpha & \frac{\alpha\,\beta\,(\beta-2\pi\,i)}{1-e^{-\alpha}}\\
         0 & -2\alpha \\
       \end{array}
     \right),
  & \hbox{if $|\alpha|<2\pi$ and $\beta\notin\{0,2\pi\,i\}$} \\[0.4cm]
                     \left(\begin{array}{cc}
         -\alpha & 0\\
         0 & -2\alpha \\
       \end{array}
     \right), & \hbox{if $\alpha\in \C$ and $\beta\in\{0,2\pi\,i\}$.}
                   \end{array}
                 \right.
\end{equation}
 The interior of $D$ is
 $\textrm{Int}(D)=\big\{(\alpha,\beta)\in \C^2:\,|\alpha|<2\pi\big\}$, and $(\alpha,\beta)\mapsto Z(\alpha,\beta)$
 is analytic here.
 On the other hand, the restriction of $Z(\alpha,\beta)$ to $\textrm{Int}(D)$ can be prolonged to the $\mathcal M$-valued map
 $$(\alpha,\beta)\mapsto
 P(\alpha,\beta):=\left(
       \begin{array}{cc}
         -\alpha & \frac{\alpha\,\beta\,(\beta-2\pi\,i)}{1-e^{-\alpha}}\\
         0 & -2\alpha \\
       \end{array}
     \right), $$
 which is analytic on the open set
 $$\Omega=\big\{(\alpha,\beta)\in \C^2:\,\alpha\neq
 2k\pi\,i\,\,\text{with $k\in \mathbb{Z}\setminus\{0\}$}\big\},$$
 and $\Omega$ clearly contains
 $\textrm{Int}(D)$.
% $Z(\alpha,\beta)$ coincides with $P(\alpha,\beta)$ on $O\cap \Omega$.
 Hence, the prolongation is singular at $(2\pi\,i,2\pi\,i)$, whereas
 the BCH series $Z(2\pi\,i,2\pi\,i)$ converges to $X(2\pi\,i)$, because
 $(2\pi\,i,2\pi\,i)\in D$.
% As another possible pathology, we see that (with the same matrices as in \eqref{CESbiagibello})
% the BCH series $Z(X(\alpha),Y(\beta))$ can be prolonged to the $\mathcal M$-valued map
% $$(\alpha,\beta)\mapsto
% \widetilde{P}(\alpha,\beta):=\left(
%       \begin{array}{cc}
%         -\ln(e^{-\alpha}) & \frac{\ln(e^\alpha)\,\beta\,(\beta-2\pi\,i)}{1-e^{-\alpha}}\\
%         0 & \ln(e^{-2\alpha}) \\
%       \end{array}
%     \right), $$
% where  $\ln z=\ln|z|+i\,\theta(z)$ with $\theta(z)\in (-\pi,\pi)$,
% which is analytic on the open set $\Omega=\big\{(\alpha,\beta)\in \C^2:\,\alpha\neq
% 2k\pi\,i\,\,\text{with $k\in \mathbb{Z}\setminus\{0\}$}\big\}$.
% Hence, the prolongation is singular at $(2\pi\,i,2\pi\,i)$, whereas
% the BCH series $Z(X(2\pi\,i),Y(2\pi\,i))$ converges to $X(2\pi\,i)$.
\end{example}
\begin{remark}\label{re.perspective}
 The above Example \ref{esempio.bello.stefano} is connected with some results in \cite{BlanesCasas}, where
 the non-convergence of the BCH series is related to the singularity of its prolongation.
 Our example shows that non-convergence cannot be directly inferred from the singularity of a prolongation;
 % whereas the argument in \cite[p.\,149]{BlanesCasas} seems to go in this direction;
 Example 1 in \cite{BlanesCasas} contains computations leading to the non-convergence of the BCH series at some singular points of its prolongation.
 In this sense,
  \cite[Example 1]{BlanesCasas} provides a non-trivial
  example where the non-convergence of the BCH series occurs
  at points where the singularity of its prolongation takes place, a phenomenon which not always happens (as
Example \ref{esempio.bello.stefano} demonstrates).
\end{remark}

 Dual to the phenomenon depicted in Example \ref{esempio.bello.stefano}, we have the following scenario,
 where
 the BCH series $Z(X,Y)$ admits a global prolongation on $\g\times \g$ (for a suitable
 real and finite-dimensional Lie algebra $\g$), but the series is not everywhere convergent.
\begin{example}\label{exa.eggertiano}
 In the real algebra $M_3(\R)$ of the $3\times 3$ matrices, consider
\begin{equation}\label{esemprioduale}
 A=\left(
     \begin{array}{ccc}
       0 & 0 & 0 \\
       0 & 2\pi & 0 \\
       0 & 0 & 1 \\
     \end{array}
   \right),\quad
 B=\left(
     \begin{array}{ccc}
       0 & 0 & 0 \\
       -2\pi & 0 & 0 \\
       0 & 0 & 0 \\
     \end{array}
   \right),\quad
  C=\left(
     \begin{array}{ccc}
       0 & 0 & 0 \\
       0 & 0 & 0 \\
       -1 & 0 & 0 \\
     \end{array}
   \right).
\end{equation}
 Their commutator relations are
 $$[A,B]=2\pi\,B,\quad [A,C]=C,\quad [B,C]=0. $$
 Hence  $\g:=\textrm{span}\{A,B,C\}$
 is a Lie subalgebra of $M_3(\R)$.
 By Proposition \ref{prop.factspos}-(a),
 the BCH series $Z(X,Y)$ converges in $\g$ for $X,Y\in \g$ close enough to the null matrix.

 In Section \ref{sec:prolongation}, we shall prove that the map $(X,Y)\mapsto Z(X,Y)$ admits
 an analytic prolongation to the whole of $\g\times \g$, by using a notable abstract result by Eggert, \cite{Eggert}.
 On the other hand, since $[A,B]=2\pi\,B$, we are entitled to apply Theorem \ref{mainteo.nonconv} and derive that
 the BCH series $Z(A,B)$ does not converge, despite its global prolongability.
\end{example}
%%%%%%%%%%%%%%%%%%%%%%%%%%%%%%%%%%%%%%%%%%%%%%%%%%%%
\section{A non-convergence result for the BCH series}\label{sec:non-convergence}
  In this section we prove
 Theorem \ref{mainteo.nonconv}, the results in Example \ref{examBiagi}
 and Proposition \ref{pro.mutuiqalrelat}.

\begin{proof}[Proof (of Theorem \ref{mainteo.nonconv}).]
 Taking for granted the notation in the statement of the theorem,
 we explicitly compute the $Z_n(X,Y)$'s.

 Since we know that $Z_1(X,Y)=X+Y$, we can suppose $n\geq 2$.
 With reference to the notation in Dynkin's presentation \eqref{notation.dynkin1},
 from $[X,Y]=vY$, $[Y,X]=-vY$ and the trivial fact $[Y,Y]=0$, one gets that
 the only possibly non-vanishing summands of
 \eqref{notation.dynkin1} are related to the indices
 for which $j_1+\cdots+j_k=1$. Thus, when $n\geq 2$,
 $Z_n(X,Y)$ coincides with the sum of the terms in formula \eqref{notation.dynkin1} where $Y$ appears precisely once.

 At the formal-power-series level of $\mathbb{Q}\langle\!\langle x,y \rangle\!\rangle$, we know
 from very classical results (see e.g., \cite[eq.\,(4.173)]{BonfiglioliFulci}) that the sum of the terms in $\sum_{n\geq 1}Z_n(x,y)$
 where $y$ appears exactly once is equal to
 $$\sum\nolimits_{k=0}^\infty \frac{(-1)^kB_k}{k!} (\Ad x)^k(y),$$
 where the $B_k$'s are the Bernoulli numbers.
 Gathering these things, by degree reasons,
 $$Z_{n+1}(X,Y)=\frac{(-1)^nB_n}{n!} (\Ad X)^n(Y),\quad \forall \,\,n\geq 1.$$
 On the other hand, from $[X,Y]=vY$ one gets $(\Ad X)^n(Y)=v^n\,Y$ for $n\geq 1$, so that
\begin{equation*}%\label{Znespliciti}
  Z_{n+1}(X,Y)=\frac{(-1)^nB_n}{n!}\, v^nY,\quad \forall \,\,n\geq 1.
\end{equation*}
   Thus, the BCH series $\sum_{n=1}^\infty Z_n(X,Y)=(X+Y)+\sum_{n=2}^\infty Z_n(X,Y)$ coincides
   with
\begin{equation}\label{Znespliciti2}
 (X+Y)+\sum_{n=1}^\infty\frac{(-1)^nB_n}{n!}\,v^nY=X+\left(\sum_{n=0}^\infty\frac{B_n}{n!}\,(-v)^n\right) Y.
\end{equation}
   Since the $B_n$'s are defined by
   $\frac{z}{e^z-1}=\sum_{n=0}^\infty \frac{B_n}{n!}\,z^n$,
   the radius of convergence of the power series
   $\sum_{n=0}^\infty \frac{B_n}{n!}\,z^n$ is $2\pi$.
   It can be proved that the power series does not converge when $|z|=2\pi$
   (for completeness reason, we furnish the proof of this in Remark \ref{rem.convberno}).

   This shows that (since $Y\neq 0$) the BCH series $\sum_{n=1}^\infty Z_n(X,Y)$ converges if and only if $|v|<2\pi$.
   As for its sum, if $|v|<2\pi$ we have
\begin{align*}
    \sum_{n=0}^\infty\frac{B_n}{n!}\,(-v)^n=
    \frac{-v}{e^{-v}-1}=\psi(v),
\end{align*}
 so that, on account of  \eqref{Znespliciti2}, we get \eqref{vanbruntvisserformula2}.
 This ends the proof of Theorem \ref{mainteo.nonconv}.
\end{proof}
%\begin{remark}\label{amodelgv}
% Examples for $X,Y$ as in Theorem \ref{mainteo.nonconv} are given by the
% null-trace matrices (see \cite{VanbruntVisser})
% in \eqref{Visserchoicematrices}
%  or by those in \eqref{CESbiagi};
%  in these cases $\mathcal{A}$ can be taken as the Banach algebra $M_2(\R)$.
%  Alternatively, $\mathcal{A}$ can be the Lie algebra of the Lie group on $\R^2$
%  whose composition law is
%  $$(x_1,x_2)(y_1,y_2)=(x_1+y_1,x_2+y_2\,e^{v x_1}),$$
%  and one can take the left-invariant smooth vector fields
%  on $\R^2$
%  $$X=\frac{\de}{\de x_1}\quad \text{and}\quad
%  Y=e^{v x_1}\,\frac{\de}{\de x_2}.  $$
% From Theorem \ref{mainteo.nonconv} we infer, for example, the non-convergence of
% the BCH series $\sum_{n=1}^\infty Z_n(X,Y)$ associated with the pair of matrices
%  $$X=\left(
%       \begin{array}{cc}
%         \pi & 1 \\
%         0 & -\pi \\
%       \end{array}
%     \right),\quad
%     Y=\left(
%       \begin{array}{cc}
%         0 & 1 \\
%         0 & 0 \\
%       \end{array}
%     \right),\qquad \text{or the pair}\quad
%     X=\left(
%       \begin{array}{cc}
%         -\pi & 0 \\
%         0 & -2\pi \\
%       \end{array}
%     \right),\quad
%     Y=\left(
%       \begin{array}{cc}
%         0 & 1 \\
%         0 & 0 \\
%       \end{array}
%     \right).$$
%\end{remark}

 Next we prove the results in Example \ref{examBiagi}.
 For $v\in\R$, let $X=X(v)$ and $Y$ be respectively
 the matrices $x=x(v)$ and $y$ in \eqref{CESbiagi}. A direct computation shows that
 $$[X,Y]=\left(
           \begin{array}{cc}
             0 & v \\
             0 & 0 \\
           \end{array}
         \right)=v\,Y,$$
         so that (with $\mathcal{A}=M_2(\R)$)
 we are entitled to apply Theorem \ref{mainteo.nonconv}.
 Hence assertion (ii) of Example \ref{examBiagi} follows directly from that theorem.
  By a direct computation we have
 $$e^Xe^Y=
 \left(\begin{array}{cc}
             e^{-v} & 0 \\
             0 & e^{-2v} \\
           \end{array}
         \right)\cdot
 \left(
       \begin{array}{cc}
         1 & 1 \\
         0 & 1 \\
       \end{array}
     \right)=
     \left(
       \begin{array}{cc}
         e^{-v} & e^{-v} \\
         0 & e^{-2v} \\
       \end{array}
     \right).$$
     Thus, the Mercator series \eqref{mercatorlogexp}
 boils down to the matrix series (be it convergent or not)
\begin{equation}\label{serielogarosamatrix}
\sum_{n=1}^\infty \frac{(-1)^{n+1}}{n}\,
 \left(
   \begin{array}{cc}
     e^{-v}-1 & e^{-v} \\
     0 & e^{-2v}-1 \\
   \end{array}
 \right)^n.
\end{equation}
 When $v=0$ this series trivially converges, and its sum is
 equal to $Y$;
 hence we can assume $v\neq 0$.
 By a direct diagonalization, we have
 $$
 \left(
   \begin{array}{cc}
     e^{-v}-1 & e^{-v} \\
     0 & e^{-2v}-1 \\
   \end{array}
 \right)
 =
 \left(
   \begin{array}{cc}
     1 & 1 \\
     0 & e^{-v}-1 \\
   \end{array}
 \right)
 \cdot
 \left(
   \begin{array}{cc}
     e^{-v}-1 & 0 \\
     0 & e^{-2v}-1 \\
   \end{array}
 \right)
 \cdot
 \left(
   \begin{array}{cc}
     1 & \frac{1}{1-e^{-v}} \\
     0 & \frac{-1}{1-e^{-v}} \\
   \end{array}
 \right).
  $$
 As a consequence,
 the series \eqref{serielogarosamatrix}
 is equal to
  $$
 \left(
   \begin{array}{cc}
     1 & 1 \\
     0 & e^{-v}-1 \\
   \end{array}
 \right)
 \cdot
 \left(
   \begin{array}{cc}
     \sum\limits_{n=1}^\infty \frac{(-1)^{n+1}}{n}(e^{-v}-1)^n & 0 \\
     0 & \sum\limits_{n=1}^\infty \frac{(-1)^{n+1}}{n}(e^{-2v}-1)^n \\
   \end{array}
 \right)
 \cdot
 \left(
   \begin{array}{cc}
     1 & \frac{1}{1-e^{-v}} \\
     0 & \frac{-1}{1-e^{-v}} \\
   \end{array}
 \right).
  $$
 Now this series is convergent if and only if
 $e^{-v}-1$ and $e^{-2v}-1$ belong to $]-1,1]$, and this is equivalent to $v\geq -\ln\sqrt2$.
 This proves assertion (i) of Example \ref{examBiagi}. Finally, assertion (iii), which is equivalent to
  $$\exp\left(
       \begin{array}{cc}
         -v & \frac{v}{1-e^{-v}}\\
         0 & -2v \\
       \end{array}
     \right) =
     \left(
       \begin{array}{cc}
         e^{-v} & e^{-v} \\
         0 & e^{-2v} \\
       \end{array}
     \right)\qquad \forall\,\, v\in\ \R,$$
 can be proved by a direct diagonalization. Finally, we provide the following

 \begin{proof}[Proof (of Proposition \ref{pro.mutuiqalrelat})]
 Numbers are related to the statements in the thesis of
 Proposition \ref{pro.mutuiqalrelat}:
\begin{enumerate}
 \item
 Sufficiency follows from statement (c) in Proposition \ref{prop.factspos}.
 The lack of necessity is shown in Example \ref{examBiagi}, if one takes
 $|v|\geq 2\pi$.

 \item
  Sufficiency follows from statement (b) in Proposition \ref{prop.factspos}.
 The lack of necessity is shown in
 Example \ref{examBiagi}, if one takes
 $v< -\ln\sqrt2$.
 Another simpler example: if one chooses $\mathcal{A}=\R$, $y=0$ and $x>\ln2$, then
 (III) holds with $V=x$, but $\sum_{n=1}^\infty (-1)^{n+1}(e^x-1)^n/n$ is not convergent.

 \item
 On the one hand, it is simple to show that (II) may hold without (I): for instance,
 in $\R$, if we take $x>\ln 2$ and $y=0$, then the BCH series boils down to $\sum_{n=1}^\infty Z_n(x,0)=x$ (and is therefore trivially convergent),
 whereas the Mercator series for $\ln(e^xe^y)$ does not converge, as observed above.
 A less trivial example is again given by Example \ref{examBiagi}, by taking $v\in (-2\pi,-\ln\sqrt2)$.
 Vice versa, the choice $v>2\pi$ yields an example for which (I) holds true but (II) is false.

 \item
 Taking the example by Wei \cite{Wei} related to the Banach algebra
 $\mathcal{M}=M_2(\R)$ and the matrices
 $$ W:=\left(
       \begin{array}{cc}
         0 & -5\pi/4\\
         5\pi/4 & 0 \\
       \end{array}
     \right)\quad \text{and}\quad
     Y:=\left(
       \begin{array}{cc}
         0 & 1 \\
         0 & 0 \\
       \end{array}
     \right),$$
 it can be proved that there does not exist any logarithm of $e^We^Y$ in $\mathcal{M}$. Thus, in view of
 statements (b) and (c) in Proposition \ref{prop.factspos}, both the BCH series $\sum_n Z_n(W,Y)$ and the Mercator series
 $\sum_n L_n(W,Y)$ cannot converge, otherwise they would provide such a logarithm.

 \item
 It can be easily seen that, with the following choice
 \begin{equation*}%\label{altroCESbiagi}
    A:=
\left(
  \begin{array}{cc}
    \ln 2 & -2\pi \\
    2\pi & \ln 2 \\
  \end{array}
\right)
\quad \text{and}\quad
     B:=\left(
       \begin{array}{cc}
         0 & 0 \\
         0 & 0 \\
       \end{array}
     \right),
 \end{equation*}
 then  (as $A,B$ commute) the BCH series boils down to $Z(A,B)=Z_1(A,B)=A+B=A$, whereas
 the Mercator series is given by
 $$
 L(A,B)=\sum_{n=1}^\infty \frac{(-1)^{n+1}}{n}\,
 \left(
   \begin{array}{cc}
     1 & 0 \\
     0 & 1 \\
   \end{array}
 \right)^n=
 \left(
   \begin{array}{cc}
     \ln 2 & 0 \\
     0 & \ln 2 \\
   \end{array}
 \right).$$
 Thus, $Z(A,B)\neq L(A,B)$ even if they are both convergent.
\end{enumerate}
 This ends the proof.
\end{proof}
%%%%%%%%%%%%%%%%%%%%%%%%%%%%%%%%%%%%%%%%%%%%%%%%%%%%%%MOMENTAN.RIPRISTINATO
% \begin{remark}
% In the case $u=c=0$, Van-Brunt and Visser's identity \eqref{vanbruntvisserformula3}
% becomes
%\begin{equation}\label{discourage}
%     \exp\bigg(X+\frac{v}{1-e^{-v}}
% Y\bigg)=e^Xe^Y.
%\end{equation}
%  Whereas there are no real singular values for this identity, in the complex case
%  one might ask what happens for $v\in \{2\pi k\,i:\,k\in \mathbb{Z}\setminus\{0\} \}$.
%  For example,  when $v=-2\pi\, i$ and (compare to \eqref{CESbiagi})
% \begin{equation*}
% X=\left(
%       \begin{array}{cc}
%         2\pi\,i & 0 \\
%         0 & 4\pi\,i \\
%       \end{array}
%     \right)\quad \text{and}\quad
%     Y=\left(
%       \begin{array}{cc}
%         0 & 1 \\
%         0 & 0 \\
%       \end{array}
%     \right),
%\end{equation*}
% identity \eqref{discourage} seems to discourage the existence
% of a logarithm for $e^Xe^Y$. This is obviously not the case, as
% $e^Xe^Y$ is invertible and any invertible matrix has a (complex) logarithm.
% For instance, since
% $e^Xe^Y=\left(
%       \begin{array}{cc}
%         1 & 1 \\
%         0 & 1 \\
%       \end{array}
%     \right)$,
% a logarithm of $e^Xe^Y$ is, e.g., $\left(
%       \begin{array}{cc}
%         0 & 1 \\
%         0 & 0 \\
%       \end{array}
%     \right)$.
%\end{remark}
%%%%%%%%%%%%%%%%%%%%%%%%%%%%%%%%%%%%%%%%%%%%%%%%%%%%%%%%%%%%%%%%%%%%%%%%%%%%%%%%%%%%%%%%%%%
\section{Prolongation issues for the BCH series}\label{sec:prolongation}
 As a first prolongation problem, we prove the assertions in
 Example \ref{esempio.bello.stefano}, showing that
 the BCH series $Z(X,Y)$ can possess an analytic prolongation
 which is singular at $(X_0,Y_0)$, whereas
 $Z(X_0,Y_0)$ converges.

\begin{proof}[Proof (of Example \ref{esempio.bello.stefano})]
   Let $X(\alpha),Y(\beta)\in \mathcal{M}=M_2(\C)$
 be as in \eqref{CESbiagibello}. Note that
 $$[X(\alpha),Y(\beta)]=\alpha\,Y(\beta).$$
 When $\beta\in \{0,2\pi\,i\}$, we have $Y(\beta)=0$ and the BCH series
 $Z(\alpha,\beta):=Z(X(\alpha),Y(\beta))$ boils down to $X(\alpha)$. Instead, when
 $\beta\notin \{0,2\pi\,i\}$, we can apply Theorem \ref{mainteo.nonconv}, deducing
 the convergence of $Z(\alpha,\beta)$
 precisely when $|\alpha|<2\pi$
 to the sum (see \eqref{vanbruntvisserformula2})
\begin{equation*}
 X(\alpha)+\psi(\alpha)\,Y(\beta)=
 \left(
       \begin{array}{cc}
         -\alpha & \frac{\alpha\,\beta\,(\beta-2\pi\,i)}{1-e^{-\alpha}}\\
         0 & -2\alpha \\
       \end{array}
     \right).
\end{equation*}
 Gathering these facts, we obtain that
 $Z(\alpha,\beta)$ converges precisely on the set $D$ defined in \eqref{CESbiagibello2} and its sum is
 \eqref{sommosa}.
 Since $\textrm{Int}(D)=\big\{(\alpha,\beta)\in \C^2:\,|\alpha|<2\pi\big\}$, and
 since
 $$(\alpha,\beta)\mapsto
 P(\alpha,\beta):=\left(
       \begin{array}{cc}
         -\alpha & \frac{\alpha\,\beta\,(\beta-2\pi\,i)}{1-e^{-\alpha}}\\
         0 & -2\alpha \\
       \end{array}
     \right) $$
 is naturally defined and analytic on the open set
 $$\Omega=\big\{(\alpha,\beta)\in \C^2:\,\alpha\neq
 2k\pi\,i\,\,\text{with $k\in \mathbb{Z}\setminus\{0\}$}\big\},$$
 we see by direct inspection that the restriction of $Z(\alpha,\beta)$ to $\textrm{Int}(D)$
 can be analytically prolonged by the function $P(\alpha,\beta)$ on $\Omega\supset \textrm{Int}(D)$.

 We note that $P$ is singular at $(2\pi\,i,2\pi\,i)$, since one has
 $$P(2\pi\,i +\e^2,2\pi\,i+\e)=
 \left(
       \begin{array}{cc}
         -2\pi\,i -\e^2 & \dfrac{(2\pi\,i+\e^2)\,(2\pi\,i+\e)\,\e}{1-e^{-\e^2}}\\[0.4cm]
         0 &  -4\pi\,i -2\e^2 \\
       \end{array}
     \right),
 $$
 and the limit of the $(1,2)$-entry as $\e\to 0$ does not exist.
 However, the BCH series $Z(2\pi\,i,2\pi\,i)$ converges since $(2\pi\,i,2\pi\,i)\in D$.
\end{proof}
 We are left to prove what is stated in Example \ref{exa.eggertiano}, showing a
 prolongation issue dual to that highlighted in Example \ref{esempio.bello.stefano}:
 we provide a Lie algebra of real matrices in which the BCH series is not everywhere
convergent, yet admitting a global prolongation.

% In \cite{Eggert} Eggert studied prolongability of the BCH operation $(X,Y)\mapsto Z(X,Y)$
% on a finite dimensional Lie algebra, as we now recall.
%\begin{remark}
% Let $\mathfrak{h}$ be a finite-dimensional real Lie algebra and let us fix any norm $\|\cdot\|$ on $\mathfrak{h}$; arguing as in
% the proof of Proposition \ref{prop.factspos}-(a), one can show that there exists
% $\varepsilon>0$ such that the BCH series $\sum_n Z_n(a,b)$ is convergent whenever $\|a\|,\|b\|<\varepsilon$.
% One can prove that
% the local operation
% $(a,b)\mapsto \sum_n Z_n(a,b)$
% defines a local Lie group.
%
% In \cite{Eggert} it is proved that this operation can be continued to the whole of $\mathfrak{h}\times \mathfrak{h}$
% if and only if the connected and simply connected
% Lie group associated with $\mathfrak{h}$ by Lie's Third Theorem is globally isomorphic to $\mathfrak{h}$
% via the exponential map.
%\end{remark}
%

\begin{proof}[Proof (of Example \ref{exa.eggertiano})]
 Let $A,B,C$ be as in \eqref{esemprioduale} and let $\g=\textrm{span}\{A,B,C\}$.
 In order to prove the existence of a global extension of the BCH series $Z(X,Y)$ to the whole of
 $\g\times \g$, we apply a general abstract result by Eggert, \cite{Eggert}.

 Indeed,
 Theorem 4.4 in \cite{Eggert} proves that the BCH operation $(X,Y)\mapsto Z(X,Y)$
 (certainly well posed when $X,Y$ are close to $0\in \g$) can be analytically continued to the whole of $\g\times \g$
 if and only if the connected and simply connected Lie group $\G$ associated with $\g$ by Lie's Third Theorem is globally isomorphic to $\g$
 via the exponential map.

 Now, by some  computations  in \cite{BonfLancCPAA}, we know that the following Lie group does exactly this job:
 $\G=(\R^3,\cdot)$ where
 $$x\cdot y=(x_1+y_1, x_2+e^{2\pi\,x_1}y_2, x_3+e^{x_1}y_3). $$
 Indeed, the Lie algebra $\textrm{Lie}(\G)$ of $\G$ is the Lie algebra of vector fields spanned by
 $$X_1:=\frac{\de}{\de x_1},\quad
 X_2:=e^{2\pi\,x_1}\,\frac{\de}{\de x_2},\quad
 X_3:=e^{x_1}\,\frac{\de}{\de x_3}, $$
 and the linear map sending $X_1,X_2,X_3$ to $A,B,C$ (respectively)
 is an isomorphism of Lie algebras between $\textrm{Lie}(\G)$ and $\g$.

 After some tedious computations (see \cite{BonfLancCPAA}), one recognizes that  the exponential map $\Exp:\textrm{Lie}(\G)\to \G$ is explicitly given by
 $$\Exp(\xi_1X_1+\xi_2X_2+\xi_3X_3)= \bigg(\xi_1, \xi_2\,\frac{e^{2\pi\,\xi_1}-1}{2\pi\,\xi_1}, \xi_3\,\frac{e^{\xi_1}-1}{\xi_1}\bigg),\quad \xi_1,\xi_2,\xi_3\in \R,$$
 which is clearly invertible, with smooth inverse map.

 Summing up, by \cite[Theorem 4.4]{Eggert} we infer that the map
 $(X,Y)\mapsto Z(X,Y)$
 can be analytically continued to the whole of $\g\times \g$. However, $Z(A,B)$ does not converge,
 as it follows by our Theorem  \ref{mainteo.nonconv}, since $[A,B]=2\pi\,B$.
\end{proof}
 Finally, somewhat connected with \cite[Theorem 4.1]{BlanesCasas} (related to an unpublished
 result by Mityagin) on the improved convergence domain $\{\|X\|+\|Y\|<\pi\}$ for the BCH series $Z(X,Y)$ in matrix Lie algebras,
 we give the next example; this shows that the the latter improved domain is not a convergence domain for the Mercator logarithm $L(X,Y)$.
\begin{example}\label{mercatorfinale}
 Consider the matrices
 $x(v)$ and $y$ in
 \eqref{CESbiagi}, for $v\in \R$. By Example \ref{examBiagi}
 we already know that
 the Mercator series $L(x(v),y)$ for $\ln(e^{x(v)}e^y)$  converges in
 the Banach algebra $\mathcal{M}$ of the $2\times 2$ real matrices if and only if $v\geq -\ln\sqrt2$.

 As a consequence, if we take $v=-1$, the Mercator series $L(x(-1),y)$ does not converge, even if
 $\|x(-1)\|+\|y\|=3<\pi$.
 Here we are considering the operator norm
 $$\|A\|=\sup_{|x|=1}|Ax|,\quad A\in \mathcal{M},$$
 and $|\cdot|$ is the standard Euclidean norm on $\R^2$.

\end{example}

%%%%%%%%%%%%%%%%%%%%%%%%%%%%%%%%%%%%%%%%%%%%%%%%%%%%%
\section{Appendix}\label{sec:appendix}
 For the sake of completeness, we give the following:
\begin{proof}[Proof of Proposition \ref{prop.factspos}.]
 We split the proof according to the statement of the proposition.\medskip

 (a). Let $\|x\|+\|y\|<\ln 2$. Then, by the compatibility of the norm of $\mathcal{A}$ with the product, we have (see \eqref{mercatorlogexp})
 \begin{gather}\label{stimaLnrappr}
 \begin{split}
    \sum_{n=1}^\infty \|L_n(x,y)\| &\leq
    \sum_{n=1}^\infty \frac{1}{n}
    \sum_{(i_1,j_1),\ldots,(i_n,j_n)\neq (0,0)}
 \frac{\|x\|^{i_1}\|y\|^{j_1}\cdots \|x\|^{i_n}\|y\|^{j_n}}{i_1!\,j_1!\,\cdots\,i_n!\,j_n!}\\
 &= \sum_{n=1}^\infty \frac{1}{n} \Big(e^{\|x\|+\|y\|}-1\Big)^n=-\ln(2-e^{\|x\|+\|y\|})<\infty.
 \end{split}
 \end{gather}
  The last equality is a consequence of   $\|x\|+\|y\|<\ln 2$ and the trivial fact $\sum_n w^n/n=-\ln(1-w)$, valid for $w\in [-1,1)$.
  This proves the absolute convergence of the Mercator series $L(x,y)$ when $\|x\|+\|y\|<\ln 2$. The above
  computation shows that we can rearrange the sums over $n$ and over $(i_1,j_1),\ldots,(i_n,j_n)$
  as we please. We group homogeneous terms of the same degree as follows:
\begin{align*}
 &L_n(x,y)=\sum_{k=n}^\infty L_{n,k}(x,y)\quad
 \text{where}\\
 &L_{n,k}(x,y):=\frac{(-1)^{n+1}}{n}
 \!\!\!\sum_{\substack{(i_1,j_1),\ldots,(i_n,j_n)\neq (0,0)\\
 i_1+j_1+\cdots+i_n+j_n=k}}\!
  \frac{x^{i_1}y^{j_1}\cdots x^{i_n}y^{j_n}}{i_1!\,j_1!\,\cdots\,i_n!\,j_n!}\,.
  \end{align*}
 This gives the computation (the sums can be interchanged due to absolute convergence)
 \begin{align*}
    L(x,y)&=\sum_{n=1}^\infty L_n(x,y)=\sum_{n=1}^\infty
    \sum_{k=n}^\infty L_{n,k}(x,y)=\sum_{k=1}^\infty \sum_{n=1}^k L_{n,k}(x,y).
 \end{align*}
 We claim that the last member is equal to the BCH series $\sum_{k=1}^\infty Z_k(x,y)$.
 This will give the equality $L(x,y)=Z(x,y)$ when $\|x\|+\|y\|<\ln 2$. The claim is a consequence
 of the following fact:
\begin{align}\label{nonassoc}
    \sum_{n=1}^k L_{n,k}(x,y)&=
    \sum_{n=1}^k
    \frac{(-1)^{n+1}}{n}
 \!\!\!\sum_{\substack{(i_1,j_1),\ldots,(i_n,j_n)\neq (0,0)\\
 i_1+j_1+\cdots+i_n+j_n=k}}\!
  \frac{x^{i_1}y^{j_1}\cdots x^{i_n}y^{j_n}}{i_1!\,j_1!\,\cdots\,i_n!\,j_n!}=Z_k(x,y).
\end{align}
 Indeed, the last equality holds true since the middle term is precisely the associative presentation of $Z_k(x,y)$ in $\K\lan\!\lan x,y\ran\!\ran$
 (the one leading to Dynkin's presentation \eqref{notation.dynkin1} after an application of the
 Dynkin-Specht-Wever map), see e.g., \cite[Sec.\,3.1.3]{BonfiglioliFulci}. We are left to prove the absolute convergence of the BCH series when $\|x\|+\|y\|<\ln 2$:
\begin{align*}
    \sum_{k=1}^\infty \|Z_k(x,y)\|&\stackrel{\eqref{nonassoc}}{\leq}
    \sum_{k=1}^\infty
    \sum_{n=1}^k
    \frac{1}{n}
 \sum_{\substack{(i_1,j_1),\ldots,(i_n,j_n)\neq (0,0)\\
 i_1+j_1+\cdots+i_n+j_n=k}}
  \frac{\|x\|^{i_1+\cdots+i_n}\|y\|^{j_1+\cdots+j_n}}{i_1!\,j_1!\,\cdots\,i_n!\,j_n!}\\
  &=
  \sum_{n=1}^\infty
    \sum_{k=n}^\infty
    \frac{1}{n}
 \sum_{\substack{(i_1,j_1),\ldots,(i_n,j_n)\neq (0,0)\\
 i_1+j_1+\cdots+i_n+j_n=k}}
  \frac{\|x\|^{i_1+\cdots+i_n}\|y\|^{j_1+\cdots+j_n}}{i_1!\,j_1!\,\cdots\,i_n!\,j_n!}\\
  &=
  \sum_{n=1}^\infty
    \frac{1}{n}
 \sum_{(i_1,j_1),\ldots,(i_n,j_n)\neq (0,0)}
  \frac{\|x\|^{i_1+\cdots+i_n}\|y\|^{j_1+\cdots+j_n}}{i_1!\,j_1!\,\cdots\,i_n!\,j_n!}\\
  &
  \stackrel{\eqref{stimaLnrappr}}{=}-\ln(2-e^{\|x\|+\|y\|}).
    \end{align*}

 (b). Suppose that $\sum_{n} L_n(x,y)$ converges in $\mathcal{A}$. Then, by
 Abel's Lemma in Banach spaces (see e.g., \cite[Lemma 5.68]{BonfiglioliFulci}), we know that
 the power series $F(t):=\sum_{n} L_n(x,y)\,t^n$ is uniformly convergent (hence continuous)
 for $t$ in $[0,1]$, and is an $\mathcal{A}$-valued analytic function on $(0,1)$.
 Now it is a standard fact to show that there exists $\epsilon>0$ such that
\begin{equation}\label{invertlocexplog}
    \exp\Big(\sum_{n=1}^\infty \frac{(-1)^{n+1}}{n} \,w^n\Big)=I+w,\quad \text{for every $w\in \mathcal{A}$ such that $\|w\|<\epsilon$.}
\end{equation}
 As a consequence, if $t$ is suitably small (so that $\|t(e^xe^y-I)\|<\epsilon$) we have
\begin{align*}
   \exp(F(t))&\,\,=\,\,\,\exp\Big(\sum_{n=1}^\infty L_n(x,y)t^n\Big)\stackrel{\eqref{mercatorlogexp}}{=}
   \exp\bigg(\sum_{n=1}^\infty \frac{(-1)^{n+1}}{n} \Big(t(e^xe^y-I)\Big)^n\bigg)\\
   &\stackrel{\eqref{invertlocexplog}}{=}I+t(e^xe^y-I)=:G(t).
\end{align*}
 Thus, $\exp\circ F$ and $G$ are two $\mathcal{A}$-valued analytic functions on $(0,1)$,
 coinciding on some small interval $(0,\epsilon')$, with $\epsilon'>0$. By Unique Continuation we infer that
 $\exp\circ F=G$ on $(0,1)$, and by continuity we get $\exp(F(1))=G(1)$. The latter identity is precisely $\exp(L(x,y))=e^xe^y$.
  \medskip

 (c). We set $F(t):=\sum_{n} Z_n(tx,ty)$. Since $Z_n$ is a homogeneous polynomial of degree $n$, we have
 $F(t)=\sum_{n} Z_n(x,y)t^n$. Arguing as above we know that the power series $F(t)$
 is uniformly convergent (hence continuous) for $t$ in $[0,1]$, and an $\mathcal{A}$-valued
 analytic function on $(0,1)$. When $t$ is small (say $t\in [0,\epsilon]$ with $\epsilon=\epsilon(x,y)>0$), we have $\|tx\|+\|ty\|<\ln 2$, hence
 (by part (a) of the proof) $F(t)$ is a logarithm of $e^{tx}e^{ty}$:
 $$\exp(F(t))=e^{tx}e^{ty}\quad \forall\,\,t\in [0,\epsilon].$$
 Now, both sides of this identity are analytic functions of $t$ on $(0,1)$, hence this identity
 is valid throughout $(0,1)$ by Unique Continuation. By continuity, the identity remains true for $t=1$:
 $$\textstyle \exp(F(1))=e^xe^y, \qquad \text{i.e.,}\quad
 \exp\Big(\sum_{n} Z_n(x,y)\Big)=e^xe^y.$$
 This is exactly what we wanted to prove.
\end{proof}
 Next, we review some special function facts for the non-convergence
 of the series associated with the Bernoulli numbers on the boundary of the disc of convergence.
\begin{remark}\label{rem.convberno}
 Let the $B_n$'s be as in Theorem \ref{mainteo.nonconv}.
 We show that the power series
 $S(z):=\sum_{n = 0}^\infty \frac{B_n}{n!}\,z^n$ is not convergent on the boundary $\{|z|=2\pi\}$
 of its convergence disc;
 since $B_{2n+1} = 0$ for every $n\geq 1$,
% inoltre, la serie di potenze complessa $S(z) = \sum_{n = 0}^\infty B_n/n!\,z^n$
% ha raggio di convergenza
%  $\rho = 2\,\pi$ e, per ogni $z\in B(0,2\,\pi)$, risulta
%  $$\sum_{n = 1}^\infty\frac{B_n}{n!}\,z^n = 1 - \frac{z}{2}+\sum_{n = 1}^\infty\frac{B_{2n}}{(2n)!}z^{2n} =
%  \frac{z}{e^z - 1}.$$
 $S(z)$ converges if and only if
 the series
 $\widetilde{S}(z):=\sum_{n = 1}^\infty \frac{B_{2n}}{(2n)!}\,z^{2n}$
 converges.
 We use a formula relating the $B_n$'s to Riemann's
 $\zeta$ function (see \cite[\S 3.16, p.\,117]{WangGuo}):
  $$\zeta(2n)= \frac{(2\,\pi)^{2n}\,(-1)^{n-1}B_{2n}}{2\cdot (2n)!}, \qquad \text{for every $n\in \mathbb{N}$}.$$
  Since $\zeta(2n)\longto 1$ as $n\to\infty$,
  %(come si riconosce facilmente
  %confrontando $\zeta(s)$ con l'integrale $\int_1^\infty 1/x^s\,\d x$ per ogni $s \in (1,\infty)$), otteniamo:
  we get
  $\dfrac{|B_{2n}|}{(2n)!} \sim \dfrac{2}{(2\,\pi)^{2n}}$  as $n\to\infty$.
  Therefore, if $|z| = 2\,\pi$, then $\widetilde{S}(z)$ cannot converge, since
  $\frac{|B_{2n}|}{(2n)!}\,|z|^{2n}\longto 2$ as $n\to \infty$.
\end{remark}
%%%%%%%%%%%%%%%%%%%%%%%%%%%%%%%%%%%%%%%%%%%%%%%%%%%%
\section*{Acknowledgements}
 The authors would like to thank the anonymous Referee of this paper whose
 constructive criticism lead us to improve a former version of the manuscript.

%%%%%%%%%%%%%%%%%%%%%%%%%%%%%%%%%%%%%%%%%%%%%%%%%%%%

\end{document}